\documentclass[conference,letterpaper]{IEEEtran}

\usepackage{amsmath,amssymb,amsfonts}
\usepackage{algorithmic}
\usepackage{graphicx}
\usepackage{textcomp}
\usepackage{xcolor}
\usepackage{changepage} 
\usepackage[bottom=1.1in, left=0.63in, right=0.63in, top=0.75in]{geometry}
\def\BibTeX{{\rm B\kern-.05em{\sc i\kern-.025em b}\kern-.08em
    T\kern-.1667em\lower.7ex\hbox{E}\kern-.125emX}}

\usepackage{amsthm}  
\usepackage{amsmath} 
\allowdisplaybreaks
\usepackage{amssymb} 
\usepackage{graphicx}
\usepackage[colorlinks=true, allcolors=blue]{hyperref}
\usepackage{times}

\newtheorem{definition}{Definition}
\newtheorem{theorem}{Theorem}
\renewcommand{\thesection}{\Roman{section}}
\newtheorem{lemma}{Lemma}
\newtheorem{prop}{Proposition}
\newtheorem{remark}{Remark}
\newtheorem{corollary}{Corollary}

\newcommand{\off}[1]{}

\begin{document}


\title{Variable Min-Cut Max-Flow\\ Bounds and Algorithms in Finite Regime\vspace{-0.3 cm}}

\off{\author{Anonymous Authors (Double Blind)\vspace{-5.cm}}}

\author{Rivka Gitik and Alejandro Cohen\vspace{-0.0cm}\\
Faculty of Electrical and Computer Engineering, Technion—Israel Institute of Technology, Haifa, Israel,\\Emails: rivkagitik@campus.technion.ac.il and alecohen@technion.ac.il\vspace{-0.6 cm}}


\maketitle
\begin{abstract}
    The maximum achievable capacity from source to destination in a network is limited by the min-cut max-flow bound; this serves as a converse limit. In practice, link capacities often fluctuate due to dynamic network conditions. In this work, we introduce a novel analytical framework that leverages tools from computational geometry to analyze throughput in heterogeneous networks with variable link capacities in a finite regime. Within this model, we derive new performance bounds and demonstrate that increasing the number of links can reduce throughput variability by nearly $90\%$. We formally define a notion of network stability and show that an unstable graph can have an exponential number of different min-cut sets, up to $O(2^{|E|})$. To address this complexity, we propose an algorithm that enforces stability with time complexity $O(|E|^2 + |V|)$, and further suggest mitigating the delay-throughput tradeoff using adaptive rateless random linear network coding (AR-RLNC).
\end{abstract}

\vspace{-0.1cm}
\section{Introduction}
The need to account for uncertainty is an inherent characteristic of many modern systems, particularly in the domains of wireless communication networks, autonomous systems, including self-driving vehicles, drone swarms, cooperative navigation, and autonomous intersection management. In these contexts, fluctuating transmission rates, unreliable links, and dynamic environments present significant challenges in ensuring the robustness and reliability of the system. Consequently, extensive research has focused on developing models and algorithms that can effectively capture, mitigate, and adapt to this variability, with the goal of improving stability, efficiency, and overall system performance \cite{huang2019evaluation, zaitseva2023review, rutkowski2016path,tang2023uncertainty}.

To address these challenges from a new perspective, we draw inspiration from computational geometry (CG), where uncertainty and variability have long been central topics of study. One of the key intersections between networking and CG lies in the use of graph-based representations. Networks are naturally modeled as graphs \cite{li2003linear,ref3,li2011linear}, and graphs also play a central role in fundamental CG problems such as Delaunay triangulations and line arrangements \cite{de2000computational}. In CG, extensive techniques have been developed to analyze structural variability, such as uncertainty models, probabilistic frameworks, and $\epsilon$-geometry \cite{agarwal2009indexing, salesin1989epsilon, gitik2022computational, gitik2021euclidean, gitik2021voronoi}. In this work, we apply these variability-focused methods to the analysis of fluctuating link behavior in networked systems.

To the best of our knowledge, this is the first work to investigate the min-cut max-flow problem in networks modeled as graphs with fluctuating link capacities in a finite regime. Our aim is to define the min-cut max-flow under this variability model formally and to compute it while addressing four key objectives: (1) maximizing throughput, (2) minimizing in-order delivery delay, (3) limiting the extent of variability mitigation, and (4) ensuring computational efficiency.

The paper is organized as follows. Section~\ref{relatedWork} provides background and reviews related work. Section~\ref{systemSection} introduces the system model, outlines its key features, and defines the setting for a single link, along with the core definitions used throughout the paper. Section~\ref{ThroughputSection} extends the analysis from a single link to a general network and derives lower and upper bounds on achievable throughput. In Section~\ref{stabilitySection}, we examine the limitations of these bounds and formally define stability as a network configuration in which throughput remains robust, i.e., when the network topology consistently preserves the same set of bottleneck links despite the variability introduced by the model.  Section~\ref{algorithmSection} discusses the challenges in defining stability and introduces an algorithm for enforcing it in practice. Section~\ref{achievabilitySection} introduces an adaptive rateless random linear network coding (AR-RLNC) scheme that can achieve either the throughput bounds or a more efficient, stability-compliant bounds. The scheme can also be adapted to implement the stability enforcement algorithm. Finally, Section~\ref{ConclusionsSection} concludes the paper and outlines directions for future research.
\section{Related Work and Novel Approach}\label{relatedWork}
Min-cut max-flow is a fundamental concept in graph theory and has been applied across various domains, including computer networks \cite{hwang1995min,nozawa1990max}, neural networks \cite{rizzi2002adaptive}, and computer vision tasks \cite{kohli2010dynamic} such as image segmentation \cite{zeng2008topology,boykov2001interactive} and energy minimization \cite{boykov2004experimental}. In the context of such applications, particularly in networked systems, a key challenge lies in accurately modeling link capacities under uncertainty. To address this, a variety of models and representations have been proposed, which can be broadly categorized into probabilistic, optimization-based, and geometrical frameworks.

Probabilistic approaches evaluate the distribution of the minimum cut capacity of graphs. Fujii and Wadayame \cite{ref4,ref5,ref6} present a coding theoretic approach for evaluating the accumulate distribution of the minimum cut capacity of random graphs. They drive the lower bound for the probabilistic capacity. The main drawback of probabilistic models is that they always underestimate the worst-case variability. 
For the optimization models, Bertsimas and Sim \cite{optimizationtBerstsimas} address data uncertainty for discrete optimization and network flow problems in a way that allows controlling the degree of conservatism of the solution. Chauhan et al. \cite{robustChaunhan} discuss a robust network interdiction problem considering uncertainties in arc capacities and resource consumption. The problem involves an adversary seeking to maximize the flow of commodity through the network, and an interdictor whose objective is to minimize this flow.
The main drawback of the optimization approach is that it often tends to become NP-hard.
For a geometrical model, Minoux \cite{polyhedralMinox} determines a robust maximum flow value in a network with uncertain link capacities taken in a polyhedral uncertainty set. Formulating a problem as a geometric one can be challenging in certain scenarios.  

Numerous studies have chosen the random graph approach \cite{ref4,ref5,ref6}. A random graph obtained by starting with a set of $n$ isolated vertices and adding successive edges between them subject to some probability. Wang et al. \cite{wang2007maximum} study the random graph with Bernoulli distributed weights and show the statistical property of the maximum flow. Karger \cite{karger1994random} shows that the sparse graph, which arises when randomly sampling the edges of a graph, will accurately approximate the value of all cuts in the original graph.

Closely related questions that emerge in the context of variability include stability, robustness, and reliability. Klimm et al. \cite{robustnessklimm} call a network topology robust if the maximal node potential needed to satisfy a set of demands never increases when demands are decreased and proposed an efficient algorithm to test robustness. Ball 
\cite{Reliabilityball}
 uses a stochastic network model in which each arc can be operative or failed. 
The network reliability analysis problem is to compute the probability that there exist operating paths from a source node to a set of destination nodes. He proves that this problem is NP-hard. Chen and Lin \cite{chen2010approximate} propose an approximate algorithm for solving network reliability analysis.  

Classical approaches in information theory for general meshed networks rely on asymptotic analysis, which assumes large blocklengths. However, such methods often fail to meet practical delay constraints. Some existing models do not provide performance bounds for heterogeneous networks \cite{dana2006capacity}, while others lack a detailed characterization of link variability \cite{wang2007maximum,karger1994random}. In addition, current literature does not offer a practical algorithm for approaching network capacity bound in a finite regime \cite{ref4,ref5,ref6}. Notably, none addresses how to achieve a throughput bound in heterogeneous networks with fluctuations while maintaining an optimal transmission rate under rateless coding.

In this paper, we propose a novel approach that considers the best-case, worst-case, and average-case scenarios, rather than relying on a probabilistic framework. In the proposed model, the link remains invariant throughout the scenario, while the link throughput exhibits variability. A disappearing link is represented by assigning it a throughput of zero, i.e., an erasure probability of one. In this sense, our model offers a broader perspective than the random graph approach, as it generalizes both the disappearance and the addition of links through changes in link throughput, making the model suited to heterogeneous networks.

\section{System Model}\label{systemSection}
This section introduces the modeling framework and provides the foundational definitions used in subsequent sections.

We consider a slotted communication setting and model the transport layer using the binary erasure channel (BEC) as a suitable abstraction to model the transport layer \cite{dias2023sliding,esfahanizadeh2024benefits}.
The forward channel is assumed to be unreliable, whereas the reverse (feedback) channel is considered reliable. Feedback errors, as exemplified in \cite{malak2019tiny}, are left for future work.

\begin{definition}
A \emph{network} is defined as an undirected graph $G = (V, E)$, where $V$ denotes the set of nodes and $E \subseteq V \times V$ denotes the set of undirected edges, also referred to as links. A single node is denoted by $v_i \in V$, and a single link is denoted by $e_i \in E$. Each link $e_i$ also serves as a reliable feedback link. The average erasure probability of the forward link $e_i$ is denoted by $p_i$. The variance associated with the forward link $e_i$ is denoted by $var_i$. 
\end{definition}

Let the round-trip time (RTT) of link $e_i$ be denoted by $\text{RTT}_i$. We define $t$ as the current time slot and let $t^- = t-\text{RTT}_i$ represent the time slot one $\text{RTT}_i$ earlier. For simplicity, we assume that the number of packets to be transmitted in each $e_i$ is at least equal to the $\text{RTT}_i$ in time slots. 

\begin{prop} \label{lem:link}
The uncertainty in the data rate of forward link $e_i$ at time $t$ is quantified using upper, lower, and mean bounds, defined as follows,
    \begin{align}
        r_i^{max}(t)&\triangleq r_i(t^-)+\sigma\underbrace{\frac{\sqrt{var_i(t)}}{\text{RTT}_i}}_{a} \label{maxRate}\\    
        r_i^{min}(t)&\triangleq r_i(t^-)-\sigma\underbrace{\frac{\sqrt{var_i(t)}}{\text{RTT}_i}}_{a} \label{minRate}\\[-0.4cm]
        r_i^{mean}(t)&\triangleq r_i(t^-), \notag
    \end{align}
    where the parameter $\sigma$ is the tunable standard deviation factor parameter, using the so-called $68 - 95 - 99.7$ rule \cite{wooditch2021normal}.
\end{prop}
\begin{proof}
The proof follows by considering the classical asymptotic regime to capture the average behavior of the channels \cite{shannon1948mathematical,shannon1956zero,cover1999elements}, and the finite regime considering the second moment of the noise to capture the variations of the channels \cite{polyanskiy2010channel,polyanskiy2010channel1,polyanskiy2011feedback} in practical modern communication systems. The available information from the past transmissions is captured by $r_i(t^-)$, which reflects the knowledge of the link prior to receiving new feedback \cite[Section V.C]{cohen2020adaptiveS}, \cite{cohen2020adaptive}. This feedback becomes available only after one $\text{RTT}_i$, introducing a delay during which the outcomes of recent transmissions remain unknown. 
The uncertainty within this feedback delay interval is quantified by the variance $var_i(t)$. Consequently, since the feedback channel is considered reliable, $r_i(t^-)$ is modeled as a Bernoulli process, and thus 
\begin{align*}
    r_i(t^-) = 1 - p_i.
\end{align*} 
Alongside, $var_i(t)$ is modeled as a binomial process comprising $\text{RTT}_i$ trials with success probability $1 - p_i$. Therefore, $var_i=\text{RTT}_i \cdot p_i(1-p_i)$.

For the underbraced term $a$ in equations \eqref{maxRate} and \eqref{minRate}, the standard deviation is divided by $\text{RTT}_i$ to normalize the $\text{RTT}_i$ transmission trials. According to \cite{polyanskiy2009dispersion}, we neglect the $O(1)$ term of the BEC.
\end{proof}

Note that for $r_i^{max}(t)$ and $r_i^{min}(t)$, setting $\sigma=1$ in equations \eqref{maxRate} and \eqref{minRate} results in a fixed value, independent of the $68 - 95 - 99.7$ rule. We will either explicitly refer to $\sigma$ or omit it, depending on the context.

As a consequence of Proposition~\ref{lem:link}, under the  $68 - 95 - 99.7$ rule, we assume that the current data rate $r_i(t)$ fluctuates within a time-varying interval, i.e., $r_i(t) \in [r_i^{min}(t), r_i^{max}(t)]$, during the $\text{RTT}_i$ period in which feedback has not yet been received. We use the notation $r_i(t_l)$ to denote a specific possible value of $r_i(t)$ at time $t$, where $l$ indexes that particular value within the interval. For notational simplicity, we omit the explicit time dependency and refer to the bounds as $r_i^{min}, r_i^{mean}$ and $r_i^{max}$. 

The following definitions provide the foundation for extending the model from individual links to a full general network.

\begin{definition} \label{def:cut}
    Let $S \subset V$ and $D \subset V$ be disjoint subsets of nodes such that $S \cap D = \emptyset$, with a source node $s \in S$ and a destination node $d \in D$. A \emph{cut} between $s$ and $d$ is defined as a subset of edges $C \subseteq E$, where each edge $e_i = (v_a, v_b) \in C$ satisfies the condition that $v_a \in S$ and $v_b \in D$. The size of a cut, denoted $|C|$, is defined as the number of edges it contains.
\end{definition}

 The number of distinct cuts is $O(2^{|V|-2})$, since each node $v_i \in V$, with $v_i \neq s$ and $v_i \neq d$, can either belong to the group containing $s$ or to the group containing $d$.

\begin{definition} \label{def:realization}
Let $r_i(t)$ denote the rate of forward link $e_i \in E$ at time slot $t$, and let $r_i(t_l)$ be a particular realization indexed by $l$. Then, $G(t_l)$ and $C(t_l)$ denote particular sets of realized rates $\{ r_i(t_l) \}$ assigned to the edges in the graph $G$ and the cut $C$, respectively, at time slot $t$.
\end{definition}

\begin{definition}
Let $G(t_l)$ denote a network and let $C(t_l)$ be a cut in $G(t_l)$, as described in Definitions~\ref{def:cut} and~\ref{def:realization}. The \emph{minimum cut}, denoted by $MC(t_l)$, is the cut with the minimum total rate across its edges, i.e., 
$$MC(t_l) \triangleq \arg\min_{C(t_l)} \left\{ \sum_{e_i \in C(t_l)} r_i(t_l) \right\}.$$
\end{definition}

Our goal is to derive performance bounds for a network operating under the described model, to characterize network variability, and to offer an adaptive rateless coding scheme that meets these bounds while maximizing throughput and minimizing in-order delivery delay, as given in the following definitions.

\begin{definition}
\emph{Throughput} is defined as the total amount of information delivered (in bits per second) in-order to the receiver through the forward links of the network $G$, divided by the total number of transmissions initiated by the sender.    
\end{definition}

\begin{definition}
\emph{In-order delivery delay} is defined as the time gap between the initial transmission of an information packet (in encoded form) by the sender and the moment it is successfully delivered and decoded in-order at the receiver, such that all preceding packets in the sequence have also been delivered and decoded.
\end{definition}

\section{Network Throughput Bounds } \label{ThroughputSection}

In this section, we extend the setting from a single link to a network. Let $G=(V,E)$ represent a network. We normalize the RTT of the network links relative to the link with the longest RTT, denoted by $e_{long} \in E$. For each link $e_i \in E$
the $\text{RTT}_i$ is divided by $\text{RTT}_{long}$, such that each edge is associated with the factor $\alpha_i$, where $\alpha _i = 1$ for $e_{long}$ and $\forall e_i \ne e_{long}, \alpha_i < 1$.

\begin{theorem} \label{the:cut}
The aggregate rate across a cut $C$ in the network $G$, at time $t$, is determined by the rates of the individual forward links that form the cut. The link rates are independent, and therefore, we obtain
\vspace{-0.2cm}
\begin{adjustwidth}{-0.05em}{0em}
\begin{align}
\small    w^{max}(C) &\triangleq  \sum_{e_i \in C}\underbrace{\frac{r_i(t^-)}{\alpha_i}}_{a} +\underbrace{{\frac{\sqrt{\sum_{e_i \in C}var_i(t)}}{\sum_{i=1}^{i=|C|} \text{RTT}_i}}}_{b} \label{cutSumMax} 
    \leq \sum_{e_i \in C}r^{max}_i \\ 
\small    w^{min}(C) &\triangleq  \sum_{e_i \in C}\underbrace{\frac{r_i(t^-)}{\alpha_i}}_{a} -\underbrace{\frac{\sqrt{\sum_{e_i \in C}var_i(t)}}{\sum_{i=1}^{i=|C|} \text{RTT}_i}}_{b} \label{cutSumMin} 
    \geq \sum_{e_i \in C}r^{min}_i  \\
\small    w^{mean}(C) &\triangleq  \sum_{e_i \in C}\frac{r_i(t^-)}{\alpha_i}. \notag
\end{align}
\end{adjustwidth}
\end{theorem}
\begin{proof}
    The underbraced term $a$ in equations~\eqref{cutSumMax} and~\eqref{cutSumMin} is the normalized transmission rate, ensuring that links with shorter $\text{RTT}_i$ values achieve proportionally higher effective rates. The numerator expression of the underbraced term $b$ in equations~\eqref{cutSumMax} and~\eqref{cutSumMin} follows the property that the variance of a sum of independent Binomial random variables is equal to the sum of their individual variances \cite[p.~53]{ross2014introduction}. The denominator serves to normalize over the total number of transmission trials $\sum_{i=1}^{i=|C|} \text{RTT}_i$, consistent with the normalization used in \eqref{maxRate} and \eqref{minRate}.
\end{proof}

For the remainder of this paper, all references to $r_i^{max}$, $r_i^{min}$ and $r_i^{mean}$, refer to their values after normalization with respect to the RTT.

Equations \eqref{cutSumMax} and \eqref{cutSumMin} reveal an interesting pattern, as the number of links in a cut increases, the magnitude of the fluctuation tends to decrease. This can be illustrated with a simple example. Consider three links, $e_1,e_2,e_3$, with rates $r_1=0.8$, $r_2=r_3=0.4$, $\text{RTT}_1=\text{RTT}_2=\text{RTT}_3=1$ and variances $var_1=1\cdot0.8\cdot0.2 = 0.16, var_2=var_3=1\cdot0.4\cdot0.6=0.24$. Suppose we wish to transmit data at a rate of 0.8 bits per second. We have two options, either use only link $e_1$, resulting in a variability of $\sqrt{0.16}/1= 0.4$, or split the traffic between $e_2$ and $e_3$, yielding a variability of $\sqrt{2\cdot0.24}/2\approx0.35 < 0.4$. Fig.~\ref{fig:cut_rate} presents a numerical illustration of this pattern as a function of the number of links in a cut $C$, evaluated for various RTT values. The observed reduction in variability, up to $90\%$, is based on this illustration and is discussed in more detail in the Appendix \ref{appendix:exmple}.

\begin{figure}
    \centering
    \includegraphics[width=1\columnwidth]{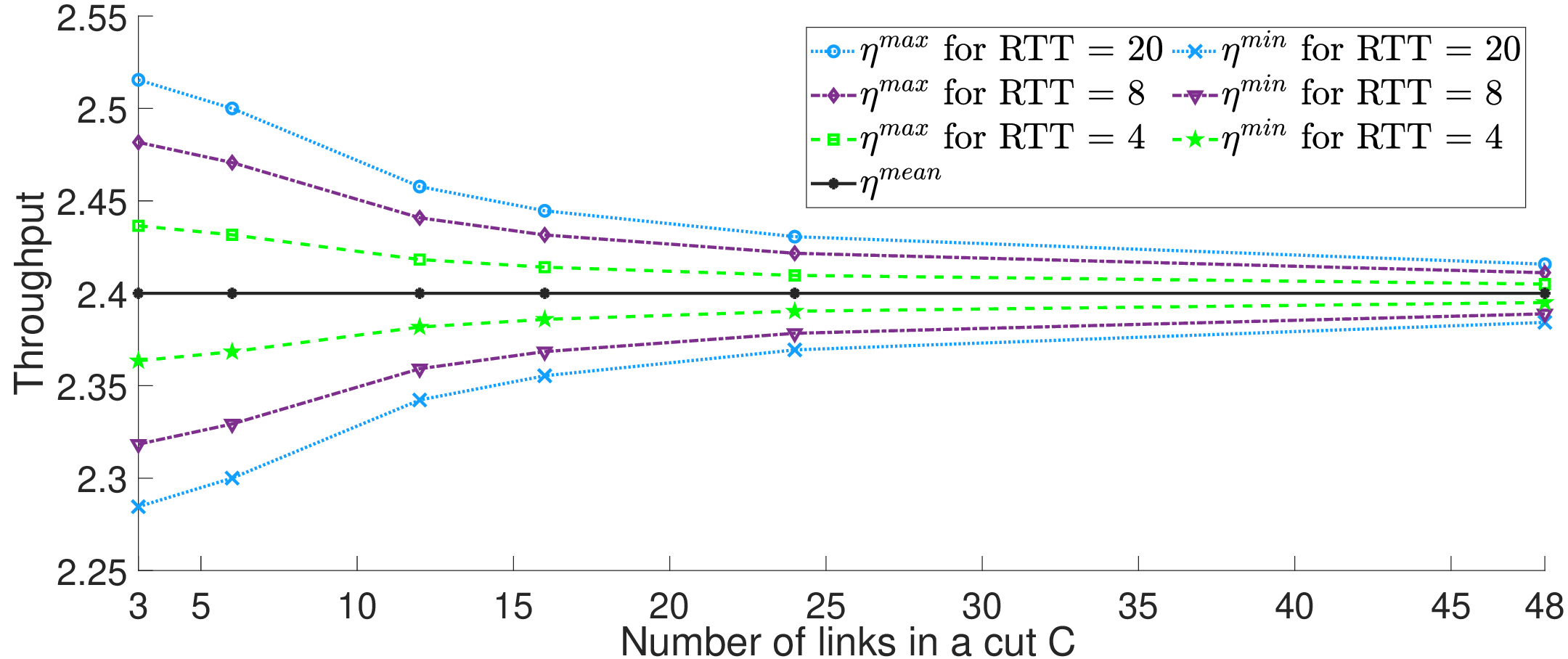}
    \caption{Throughput as a function of links in a cut $C$ for different RTT values. The target data transmission rate is $2.4$ bits per second. The RTT assigned to each link is 4, 8, and 20 for the green, purple and blue curves, respectively. The solid line is $\eta^{mean}$, while the upper and lower dashed lines illustrate $\eta^{max}$ and $\eta^{min}$, respectively, for the different RTT values. As the number of links increases or the RTT decreases, throughput fluctuations diminish, leading to tighter bounds and enhanced performance guarantees.}
    \label{fig:cut_rate}
    \vspace{-0.5cm}
\end{figure}

\begin{lemma}\label{Lemma:maxLink}
    Given a fixed rate budget for an application across a network cut, and assuming that all links across the cut have identical erasure probability $p_i$, the optimal strategy is to select the maximum number of links whose combined usage satisfies the rate constraint.
\end{lemma}
\begin{proof}
    Following Section~\ref{systemSection}, we have $\sum_{e_i \in C} \mathrm{var}_i(t) = \sum_{e_i \in C} \text{RTT}_i \cdot p_i(1 - p_i)$. Let $ x = \sum_{e_i \in C} \text{RTT}_i $. Then, the size of the uncertainty interval is 
    $$\small 2 \cdot \frac{\sqrt{x \cdot p_i(1 - p_i)}}{x}.$$
    Since $ 0 \leq p \leq 1 $, it follows that $ p_i(1 - p_i) \in [0, 1/4] $. And
    $$\small 2 \cdot \frac{\sqrt{x \cdot p_i(1 - p_i)}}{x} = 2 \cdot \sqrt{p_i(1 - p_i)} \cdot \frac{1}{\sqrt{x}}.$$
    As $ x \geq 1 $ by definition, and $ \sqrt{p_i(1 - p_i)} \geq 0 $, it follows that the expression is nonnegative. Moreover, the function $ \frac{1}{\sqrt{x}} $ is strictly decreasing as $ x \rightarrow \infty $. Thus, unless $ p_i$ is $0$ or $1$, the size of the uncertainty interval strictly decreases as $x$ increases.
\end{proof}
As established in Lemma~\ref{Lemma:maxLink}, utilizing a larger number of links enables the system to provide better performance guarantees, as the resulting bounds become tighter and more accurate in a practical finite regime \cite{cohen2021bringing,esfahanizadeh2024benefits}. The following corollary is a direct consequence of this result.
\begin{corollary}
For a given cut $C$, the maximum achievable throughput $w^{\text{max}}(C)$ converges to the mean throughput $w^{\text{mean}}(C)$ as the $\text{RTT}_i$ or the number of links in the cut increases. Specifically, $$w^{\text{max}}(C) \to w^{\text{mean}}(C),$$ with a convergence rate of $O\left(\frac{1}{\sqrt{\sum_{e_i \in C} \text{RTT}_i}}\right)$.
\end{corollary}
We now present the main theoretical bounds applicable to a general network setting.
\begin{theorem}\label{theorem:generalNetworkBound}
Let $G$ be a network. The uncertainty in the throughput, $\eta$ at time $t$ can be characterized by upper, lower, and mean bounds, which are determined by the rates of the  forward links in the minimum cut of the network, i.e., for $\min\{C \subseteq E\}$. Specifically, we have
\vspace{-0.5cm}
\begin{adjustwidth}{-0.6em}{0em}
\setlength{\belowdisplayskip}{1pt}
\begin{align}
    \eta^{max} &\triangleq  \min_{C \subseteq E}(w^{max}(C)), \label{etaMax}\\
    \eta^{min} &\triangleq  \min_{C \subseteq E}(w^{min}(C)), \label{etaMin}\\
    \eta^{mean} &\triangleq  \min_{C \subseteq E}(w^{mean}(C)). \label{etaMean}
\end{align}
\end{adjustwidth}
\end{theorem}
\begin{proof}
    The correctness of this theorem is a direct consequence of Proposition~\ref{lem:link} and Theorem~\ref{the:cut}. 
\end{proof}

Notice that in Theorem~\ref{theorem:generalNetworkBound}, each bound may be determined by a different network cut, that is, by a distinct set of edges. This observation reveals a structural complexity in the bound, which we explore further in the next section by analyzing the number and variability of such minimal cut sets.

\section{Stability of Minimum Cuts under Link Variance}
\label{stabilitySection}

This section addresses a key limitation of the bound established in Section~\ref{ThroughputSection}, the sets of edges defining the minimum cuts for $\eta^{\max}$, $\eta^{\min}$, and $\eta^{\mathrm{mean}}$, as given in equations~\eqref{etaMax}, \eqref{etaMin}, and \eqref{etaMean}, may differ due to fluctuations in link capacities. We aim to identify conditions under which these minimum cut sets coincide. Focusing on this stable case allows the analysis to be centered on a single, consistent bottleneck across all three bounds. Accordingly, this section is divided into two subsections, the first shows that the number of distinct minimum cut sets can be exponential, and the second proposes sufficient conditions under which a single cut set governs all three bounds.

\subsection{Instability of Minimum Cuts under Link Variance}\label{subSec:problem}

The notion of stability introduced in this section is motivated by analogous concepts in computational geometry, where it has been studied in the context of Voronoi diagrams~\cite{gitik2021voronoi, gitik2022computational} and Euclidean minimum spanning trees~\cite{gitik2021euclidean, gitik2022computational}.

\begin{definition}
Let $G=(V,E)$ define a network. The sets $MC(t_l)\subseteq E$ and $MC(t_m )\subseteq E$, with $l \ne m$, are considered \emph{equivalent} iff they contain exactly the same edges. That is, $ \forall e_i \in MC(t_l)$ holds that $e_i \in MC(t_m)$. And vice versa. Otherwise, the sets $MC(t_l)$ and $MC(t_m)$ are \emph{distinct}.
\end{definition}
\begin{theorem}
A network $G$ may admit an exponential number of distinct minimum cut realizations under different rate assignments.
\end{theorem}
\begin{proof} 
By constructing an example. Suppose that there are $n$ paths for the source $s$ to the destination $d$, as described in Fig.~\ref{fig:Fig1}. The number of edges is $|E|=2n$. For simplicity, assume that $[r_i^{min},r_i^{max}]$ is identical for $\forall e_i \in E$. Let all the $n$ edges originate at $s$, that is, $e_i=(s,v_i)$ at time $t$ had the rate $r_i(t_l) =  (r_i^{min}+r_i^{max})/2$. And let define the rate $r_j(t_l ) = r_i^{min}$ for any subset of the $n$ edges of the form $e_j=(v_j,t)$ and $r_j (t_l) = r_i^{max}$ to the remained edges. In this way, the min-cut set $MC(t_l)$ contains all the edges with the rate $r_i^{min}$ and all the edges with the rate $r_i(t_l) = (r_i^{min} + r_i^{max})/2$ that correspond to the path with the rate $r_i^{max}$. Consequently, graph $G$ contains $2^n$ different min-cut.
\end{proof}

\begin{figure}
\centering{\includegraphics[width=0.4\linewidth]{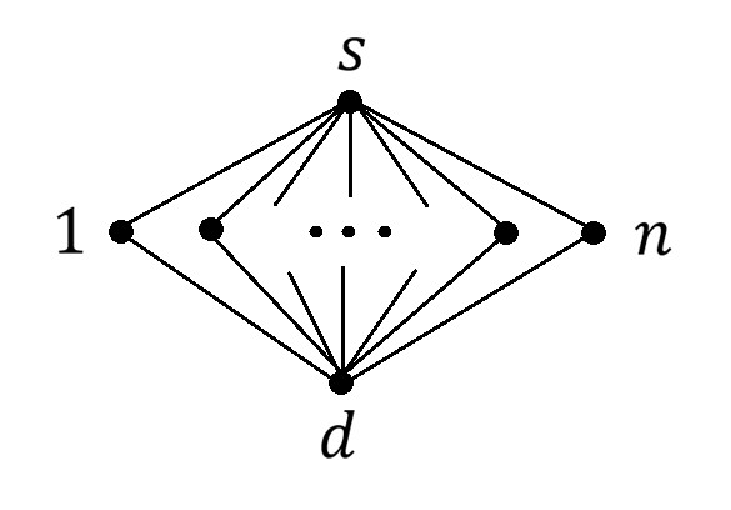}}
\caption{A network originated at $s$, the source vertex, with $n$ paths toward the destination $d$.}\vspace{-0.5cm}
\label{fig:Fig1}
\end{figure}

\subsection{Conditions for Minimum Cut Stability}

Having seen that the number of distinct minimum cut sets can be exponential, we now seek to characterize when these sets align.

\begin{definition}\label{def:stability}
Let $G=(V,E)$ define a network. The min-cut set is \emph{stable} iff
$\forall t_l,\forall t_m$, $l \ne m$, $MC(t_l)$ is equivalent to $MC(t_m)$. When the cut is stable, it is denoted as $MC \subseteq E$.
\end{definition}

Building on the preceding theorem and definitions, we now establish the following result.

\begin{theorem}
Assume $G=(V,E)$ defines a network. When $MC$ is stable, under Definition~\ref{def:stability}, the maximum and minimum throughput, between source $s$ and destination $d$, can be obtained by $e_i \in MC$ as follows,
\setlength{\jot}{1pt}
\begin{align}
\eta^{max}_{stable} &\triangleq \sum_{e_i \in MC} \frac{r_i(t^-)}{\alpha_i} + \frac{\sqrt{\sum_{e_i \in MC}var_i(t)}}{\sum_{e_i \in MC}{\text{RTT}_i}},  \label{etaMaxStable} \\
\eta^{min}_{stable} &\triangleq \sum_{e_i \in MC} \frac{r_i(t^-)}{\alpha_i} - \frac{\sqrt{\sum_{e_i \in MC}var_i(t)}}{\sum_{e_i \in MC}{\text{RTT}_i}}. \label{etaMimStable}
\end{align}
\end{theorem}

\begin{proof}
Let $e_i \in MC$ be any edge in the minimum cut, and consider any other edge $e_j \in E$ that lies on the same path from the source $s$ to the destination $d$. By the definition of stability, $e_j$ cannot replace $e_i$ in the minimum cut. This implies that $r_i(t_l) \leq r_j(t_m), \forall t_l, \forall t_m$, that is, $r_i^{\max} \leq r_j^{\min}$. By the max-flow min-cut theorem \cite{yeung2008chapter18}, the maximum achievable throughput under this condition is given by $\eta^{\max}_{\mathrm{stable}}$. 

Moreover, $r_i^{\min} < r_i^{\max} \leq r_j^{\min}$ ensures, by the same reasoning, that $e_i$ consistently limits the flow, implying that the minimum achievable throughput is $\eta^{\min}_{\mathrm{stable}}$.
\end{proof}

The key difference between equations \eqref{etaMax}, \eqref{etaMin} and equations \eqref{etaMaxStable}, \eqref{etaMimStable} lies in the definition of the underlying min-cut sets. In the latter case, $\eta^{max}_{stable}$ and $\eta^{min}_{stable}$ are computed over an identical and stable set of min-cut edges. In contrast, $\eta^{max}$ and $\eta^{min}$ may be derived from differing min-cut edge sets.

Next, we define an efficient algorithm to verify the stability of $MC$. 
\begin{theorem}
    Verifying the stability of a network $G = (V, E)$ can be done using the following algorithm, with a time complexity of $O(|E|^2)$.
\end{theorem}

\begin{proof}
    For each edge $e_i \in MC$, compare it with every other edge $e_j \in E$ that lies on the same path from the source $s$ to the destination $d$. Check whether the condition $r_i^{\max} \leq r_j^{\min}$ holds. If this condition is satisfied, then the edge $e_j$ cannot replace $e_i$ in the set $MC(t_l)$ for $\forall t_l$, as it implies $r_i(t_l) \leq r_j(t_m),\forall t_l, \forall t_m$. Since each edge $e_i \in E$ is compared with $O(|E|)$ other edges, the overall time complexity is $O(|E|^2)$.
\end{proof}

 Relying on the stable case enables the analysis to focus on a single, consistent set of edges, allowing link fluctuations to be measured exclusively within this fixed set. However, the applicability of these results is limited by the restrictive nature of the stability criterion in complex network topologies, which complicates practical implementation; we address this challenge in the next section.

\section{Stability Forcing Algorithm} \label{algorithmSection} 

The results in Section~\ref{stabilitySection} face two key limitations. The stability condition based on the $MC$ is often overly restrictive in general networks, and the $MC$ may fluctuate from one time slot to the next. To address these challenges, we next propose a stable solution framework for general network settings that identifies the minimum cut while (1) maximizing throughput, (2) minimizing the necessary reduction in link uncertainty, and (3) maintaining low computational complexity. This section comprises two subsections, the first presents the algorithm and the second analysis its runtime and properties.

\subsection{Algorithm Overview}\label{sec:algorithm-overview}

Consider a single path from $s$ to $d$ with $m$ hops. Denote the edge with the minimum mean rate by $e_1$. Suppose first that there are no two links with the same minimum mean rate. Denote the rest of the edges in ascending order $\{e_2, \dots ,e_m\}$ according to $r_i^{min}$ in a normalized equation \eqref{minRate}, that is, the minimum rate is $r_2^{min}$ and the maximum is $r_m^{min}$. We compare $r_1^{max}$ and $r_2^{min}$, when $r_1^{max} \leq r_2^{min}$, the edges $e_1$ and $e_2$ are stable, and so is the path, as $r_1^{max} \leq r_2^{min} \leq \dots \leq r_m^{min}$. Otherwise, when $r_1^{max} > r_2^{min}$ we tune the parameter $\sigma$ in normalized equations \eqref{maxRate} and \eqref{minRate} to force $r_1^{max} \leq r_2^{min}$ as follows. 

Denote by $\sigma_1$ and $\sigma_2$, the $\sigma$ belongs to links $e_1$ and $e_2$ respectably. We divide the range between $r_1^{mean}$ and $r_2^{mean}$ equally to $e_1$ and $e_2$ (see Fig.~\ref{fig:Fig2}), so that $r_1^{max}$ ends where $r_2^{min}$ starts, 
$$\small \frac{r_1^{mean} +r_2^{mean}}{2}=r_2^{mean}-\sigma_2 \frac{\sqrt{v_2(t)}}{\text{RTT}_2},$$ 
which gives
$$\sigma_2 = \frac{(r_2^{mean}-r_1^{mean})\text{RTT}_2}{2\sqrt{v_2(t)}}{},$$
and 
$$\small r_2^{min}=\frac{r_2(t^-)}{\alpha_2}-\frac{r_2^{mean}-r_1^{mean}}{2}.$$
Similarly, 
$$\small r_1^{max}=\frac{r_1(t^-)}{\alpha_1}+\frac{r_2^{mean}-r_1^{mean}}{2}.$$

\begin{figure}
    \centering
    \includegraphics[width=0.35\linewidth]{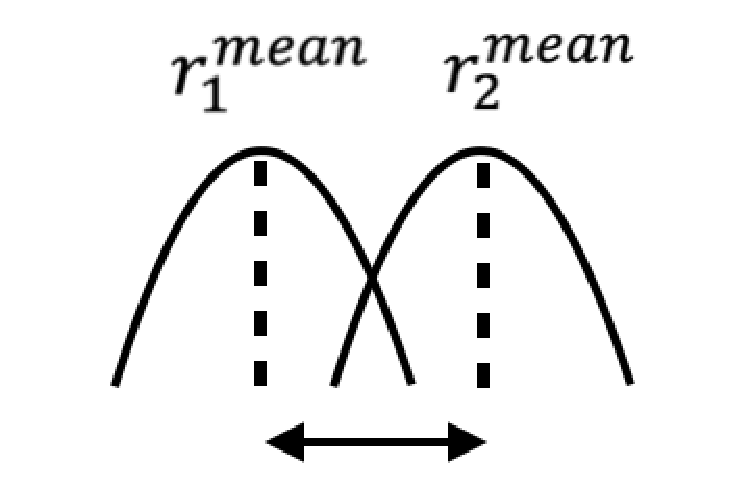}
    \caption{\label{fig:Fig2}The rates of links \( e_1 \) and \( e_2 \) are depicted, with dashed lines representing the mean rates \( r_1^{\text{mean}} \) and \( r_2^{\text{mean}} \). The distributions are shown by the arcs, and the arrow indicates the range between \( r_1^{\text{mean}} \) and \( r_2^{\text{mean}} \).}
    \vspace{-0.5cm}
\end{figure}
Next, we test whether the updated $r_1^{max}$ is smaller than $r_3^{min}$ and if $r_1^{max} > r_3^{min}$ we adjust $\sigma_3$ so that $r_1^{max} \leq r_3^{min}$ holds. And so on, we test each edge on the path $e_i \in \{e_3 ,\dots,e_m \}$. Consequently, it holds $r_1^{max} \leq r_i^{min},2 \leq i \leq m$, and the path is stable. 

Consider now the case where there are two, or more, links with the same minimum mean rate. If all of them had the same variance, we chose one of them randomly as $e_1$. Otherwise, we consider all variance of the minimum mean rate and construct a virtual edge with $r_1^{max}$ according to the smallest variance and $r_1^{min}$ according to the largest variance and define it as $e_1$. In both cases, we compare $e_1$ with the other links on the path excluding the links with the same minimum mean rate.

For the general case of a graph $G = (V, E)$, we begin by identifying the minimum cut corresponding to the average rates $r_i^{mean}$ for all $i \in \{1, \dots, |E|\}$, that is $MC(t^-)$. Then, for each edge in the cut, $e_j = (v_a, v_b) \in MC(t^-)$, we trace the path from $s$ to $v_a$, and from $v_b$ to $d$, collecting all the edges along both subpaths. We then apply the single-path procedure, as described earlier in this section, to the resulting set of edges, including $e_j$. Note that, by construction, the edge $e_j = (v_a, v_b) \in MC(t^-)$ plays the role of $e_1$ in the single-path case from $s$ to $d$.

When an edge $e_i \in E$, $e_i \notin MC(t^-)$ is on more than one path from $s$ (or $d$) to the $MC(t^-)$ edges denoted by $e_{j1}, e_{j2} \dots e_{jm}$, where $m$ is the number of paths that $e_i$ lead to them. Then, $e_i$ compared to each of thous $m$ edges of $MC(t^-)$, and associated with $\sigma_i$ which gives the maximal $r_i^{min}$. This way $r_i^{min} \geq r_{j1}^{max}, \dots, r_i^{min} \geq r_{jm}^{max}$.  
\vspace{-0.2cm}
\subsection{Runtime and Analytical Properties}
In this subsection, we analyze the complexity of the proposed algorithm and examine its properties.

\begin{theorem}
    The time complexity of the stability forcing algorithm is $O(|E|^2+|V|)$ given the $MC(t^-)$.
\end{theorem}
\begin{proof}
    Performing two depth-first search (DFS) \cite{cormen2009introduction}, one from $s$ to $\forall v_a, (v_a,v_b) \in MC(t^-)$ and the second from $d$ to $v_b,(v_a,v_b) \in MC(t^-)$ take $O(|V|+|E|)$ time. For each $e_i \in MC(t^-)$ find the edge $e_j$ with minimal $r_j^{min}$ among all the edges that originate in $s$ passing through $e_i$ and reaching $d$ overall, $\forall e_i \in MC(t^-)$ takes $O(|E| \cdot |MC(t^-)|)$ time. Adjusting the $\sigma$ parameter of the edges takes $O(|E| \cdot |MC(t^-)|)$ time. And $O(|E|\cdot|MC(t^-)|) \leq O(|E|^2)$.    
\end{proof}

The following theorem provides a worst-case analysis of the potential throughput degradation introduced by the stability-forcing algorithm.  

\begin{theorem} \label{theorem:worstCase}
    In the worst-case scenario, the algorithm's maximum achievable throughput $\eta^{max}$ asymptotically converges to the mean throughput $\eta^{mean}$; that is,
    \[
        \eta^{max} \to \eta^{mean}.
    \]
\end{theorem}
\vspace{-0.3cm}
\begin{proof}
    Consider all edges define $MC(t^-)$ and suppose $\forall e_i \in MC(t^-)$, $\exists e_j$ on the path from $s$ to $d$ through $e_i$, such that $r_j^{mean} + \epsilon = r_i^{mean}$. As $\epsilon \to 0$, $r_i^{max} \to r_i^{mean}$.
\end{proof}
\vspace{-0.2cm}
\begin{remark}\label{remark:asymptotic}
In typical settings, the asymptotic regime is modeled using a Bernoulli process that converges to the first moment \cite{shannon1948mathematical,shannon1956zero,cover1999elements}. However, short block-length regimes often incorporate second-moment effects \cite{polyanskiy2010channel,polyanskiy2010channel1,polyanskiy2011feedback}, in addition to the Bernoulli model, due to increased sensitivity to variations in channel conditions. It is important to note that the stability-forcing algorithm presented in this section appears to primarily suppress the impact of the second moment, while preserving the first-moment characteristics of the system; specifically, it reduces link uncertainty without altering the average rate $r_i^{\text{mean}}$ on any link.
\end{remark}
\begin{remark}
    As established in Corollary~\ref{Lemma:maxLink}, when either the $\text{RTT}_i$ or the number of links in the min-cut increases, the maximum achievable throughput $\eta^{max}$ converges to the mean throughput $\eta^{mean}$. This convergence justifies the applicability of Theorem~\ref{theorem:worstCase} and Remark~\ref{remark:asymptotic} even in practical, finite network scenarios.
\end{remark}

The following theorem establishes the condition under which the computed minimum cut remains unchanged, thereby ensuring the validity of the throughput bounds defined in equations~\eqref{etaMaxStable} and~\eqref{etaMimStable}, which rely on this cut. If the condition is violated, recomputation of the minimum cut is required.

\begin{theorem}
Let $P = \{e_{p1},\dots,e_{pm}\}$ indicate all the $m$ edges of a single path from $s$ to $d$ throw $e_j\in MC(t^-)$, that is, for $e_j=(v_a,v_b)$, all the edges on the path from $s$ to $v_a$ and all the edges on the path from $v_b$ to $d$. The $MC(t^-)$ remains the same as long as we have $e_{pi}^{mean} \geq e_j^{mean},\forall e_{pi} \in P$ for each $e_j\in MC$.
\end{theorem}
\begin{proof}
First, consider the case where $e_{pi}^{mean} \geq e_j^{mean},\forall e_{pi} \in P$. In this scenario $e_j$ corresponds to the edge denoted by $e_1$ in the algorithm, and $\forall e_{pi} \in P$ are mitigating compared to it. Thus, no recomputation of the minimum cut is required.
Now, consider the case where $\exists e_{pi}$ such that $e_{pi}^{mean} < e_j^{mean}$. In this situation $e_j^{mean}$ can no longer be considered as bottleneck in $MC(t^-)$. Therefore, the minimum cut must be recomputed to reflect this change.
\end{proof}

To reduce the computational complexity of computing $MC(t^-)$, a nearly linear-time approximation algorithm, such as the one proposed in \cite{spielman2004nearly}, may be used, offering significant advantages in scenarios with frequent minimum cut updates.

\section{Achievability} \label{achievabilitySection}
We consider an adaptive rateless random linear network coding (AR-RLNC) \cite{yang2014deadline,dias2023sliding,esfahanizadeh2024benefits} as a means to approach the performance bounds established in Sections~\ref{ThroughputSection} and \ref{stabilitySection}. The bounds of Section \ref{stabilitySection} can be attained through the algorithm presented in Section~\ref{algorithmSection}. We begin with a brief overview of RLNC \cite{ho2006random}, followed by a concise introduction to AR-RLNC, and conclude with a discussion on how the bounds can be effectively applied to AR-RLNC.

A major advantage of RLNC and AR-RLNC is that acknowledgments are sent once all packets (or all packets in a group - for AR-RLNC) are received, rather than for each individual packet. This is especially important in high-frequency wireless communications, where the cost of coordinating between the transmitter (TX) and receiver (RX) is considerably high \cite{dahhani2019association}.

For simplicity, we assume that the data sent from source $s$
to destination $d$ consists of $N$ packets of the same size, i.e., $\{P_1,\dots,P_N\}$. At a $t$-th time step, $s$ transmits a coded packet $E_t$. The receiver $d$ may acknowledge $s$ by sending an acknowledgment (ACK). Each encoded packet $E_t$ is a random linear combination of the original uncoded packets,
\begin{align}
    E_t = \sum_{j=1}^{N}\rho_{t,j}P_j, \notag
\end{align}
where $\rho_{t,j},$ is a coefficient defined by $ t \in \{1,2,\dots\}, j \in \{1,\dots,N\}$ randomly sampled. Ones $N$ coded packets received, $s$ can decode the original packets.

This approach imposes a large latency on the system, as for decoding the first packet, at least $N$ coded packets should arrive. The AR-RLNC addresses this issue by splitting the packets $\{P_1,\dots,P_N\}$, into $\mu$ groups, each group containing $n$ packets, that is, $\{P_1,\dots,P_n\},\{P_{n+1},\dots,P_{2n}\},\dots,\{P_{(\mu-1)n+1},\dots,P_{\mu n}\}$. Notice that $\mu \cdot n = N$ and when $\mu\cdot n \neq N$ a zero padding technique can be applied. The packets in each group $I = \{p_{(i-1)n + 1},\dots,p_{in}\},i \in \{1,\dots,\mu\}$ are encoded similarly to RLNC and transmitted in order. Let $E_t^i$ be the $I$ group encoded in time $t$ then
\setlength{\belowdisplayskip}{4pt}
\begin{align}
    E_t^i = \sum^{n}_{j=1}\rho^i_{t,j}P_{(i-1)n + j}. \notag
\end{align}
In this way $s$ can recover the whole group per receipt of the $n$ encoded packet, $n << N$. If $s$ does not receive an acknowledgment showing that the coded group $I$ received, at the end of the $n$-th transmission, it starts sending other $n$ coded packets of $I$ (with different coefficients).

Let $l$ denote the length of an encoded packet $E_t^i$, and let $\delta$ represent the associated metadata of the packet. In general, each group of packages requires $(l+\delta)n$ bits.

In-order delivery delay and throughput are interdependent. The higher the total amount of information delivered, the longer the delay. This concept is broadly explored in  \cite{cohen2020adaptiveS,cohen2020adaptive}. The AR-RLNC network coding offers a mitigation of the delay-throughput tradeoff. We propose a further improvement by selecting the parameter $\mu$ such that $(l+\delta)n$ equals $\eta^{\max}$, $\eta^{\min}$, or $\eta^{\text{mean}}$. And even more effective mitigation by setting $(l+\delta)n$ to $\eta^{\max}_{\text{stable}}$, or $\eta^{\min}_{\text{stable}}$. The latter approach yields superior performance, as it allows for targeting a fixed number of links and better represents the current network conditions. 

When $\eta^{\text{max}}$ or $\eta^{\max}_{\text{stable}}$ are selected, $n$ tends to be large; the transfer is rapid, packet loss is high, and in-order delivery delay is long. Conversely, choosing $\eta^{\text{min}}$ or $\eta^{\min}_{\text{stable}}$ results in a smaller $n$; the transfer is slow, packet loss is low, and in-order delivery delay is short.

In RLNC and AR-RLNC, feedback can be interpreted as a stopping time signal, that is, once the packets are successfully decoded, there is no need to transmit further linear combinations of those same packets. Until the sender receives an acknowledgment, it continues to transmit information. Since there is no explicit termination signal (i.e., no designated end-of-packet indicator), this coding scheme falls into the category of variable-length feedback (VLF) codes, as defined by Polyanskiy et al. in \cite{polyanskiy2011feedback}. 

\section{Conclusions and future work} \label{ConclusionsSection}

In this work, we characterized network throughput across three scenarios, best case ($\eta^{\text{max}}$), worst case ($\eta^{\text{min}}$), and average case ($\eta^{\text{mean}}$). Our analysis demonstrated that increasing the number of transmitting links can reduce throughput variability by nearly $90\%$. We introduced the concept of stability within the min-cut max-flow framework under link throughput variability and showed that an unstable network can yield up to $O(2^{|E|})$ distinct min-cut max-flow configurations. To address this, we proposed an efficient stability-forcing algorithm with time complexity $O(|E|^2 + |V|)$, which ensures that the maximum achievable throughput converges to the average in the worst case, i.e. $\eta^{\text{max}} \rightarrow \eta^{\text{mean}}$. Finally, we presented an adaptive rateless random linear network coding (AR-RLNC) scheme that effectively approaches the theoretical bounds of network throughput.

Future work includes deriving the converse limits of the model. Another direction involves extending the unicast results to the multicast setting and applying the adaptive causal network coding with feedback (AC-RLNC) \cite{cohen2020adaptiveS,cohen2020adaptive} to the variable network model. This involves analyzing the resulting throughput and delay performance bounds, as well as evaluating the effectiveness of the approach through simulation.

\bibliographystyle{IEEEtran}
\bibliography{ref_short.bib}

\appendices

\section{Further Details on the Example in Fig.~\ref{fig:cut_rate}} \label{appendix:exmple}

Consider $\text{RTT} = 4$ and compare the throughput variation between the case of 3 links and the case of 48 links. The target transmission rate is $2.4$ bits per second. For 3 links, each link must achieve a success rate of $1 - \epsilon = 0.8$. For 48 links, each link requires a success rate of $1 - \epsilon = 0.05$ to maintain the same total throughput.

\underline{\textbf{Case A}}. For 3 links and $\epsilon = 0.2$:
\begin{align*}
\eta^{\max} &\approx 2.4 + \frac{\sqrt{3 \cdot 4 \cdot 0.2 \cdot 0.8}}{3 \cdot 4} \approx 2.5155, \\
\eta^{\min} &\approx 2.4 - \frac{\sqrt{3 \cdot 4 \cdot 0.2 \cdot 0.8}}{3 \cdot 4} \approx 2.2845.
\end{align*}
Thus, the variability range is approximately $0.2308$.

\underline{\textbf{Case B}}. For 48 links and $\epsilon = 0.95$:
\begin{align*}
\eta^{\max} &\approx 2.4 + \frac{\sqrt{48 \cdot 4 \cdot 0.95 \cdot 0.05}}{48 \cdot 4} \approx 2.4157, \\
\eta^{\min} &\approx 2.4 - \frac{\sqrt{48 \cdot 4 \cdot 0.95 \cdot 0.05}}{48 \cdot 4} \approx 2.3843.
\end{align*}
The variability range is approximately $0.0314$.

As a result, the throughput variability is significantly reduced, by nearly $90\%$, when increasing the number of links from 3 to 48. This illustrates how increasing the number of links leads to more stable and predictable performance.

\end{document}